\documentclass[lettersize,journal]{IEEEtran}
\IEEEoverridecommandlockouts
\usepackage{cite}
\usepackage{amsmath,amssymb,amsfonts}
\usepackage{graphicx}
\usepackage{textcomp}
\usepackage{xcolor}
\usepackage{todonotes}
\usepackage{multirow}
\usepackage{enumitem}
\usepackage{caption}
\usepackage{algorithm}
\usepackage{algpseudocode}
\usepackage{comment}
\usepackage{amsthm}
\usepackage{hyperref}
\usepackage{booktabs}
\usepackage{array}
\usepackage[most]{tcolorbox}
\usepackage{tabularx}
\usepackage{minted}
\usepackage{stmaryrd}
\usepackage{changepage}
\usepackage{orcidlink}
\begin{document}

\newcommand{\ab}[1]{\textcolor{blue}{\textbf{AB:} #1}}
\renewcommand{\todo}[1]{\textcolor{red}{\textbf{TODO:} #1}}
\def\BibTeX{{\rm B\kern-.05em{\sc i\kern-.025em b}\kern-.08em
    T\kern-.1667em\lower.7ex\hbox{E}\kern-.125emX}}
\newcommand{\aren}[1]{\todo[color=yellow]{#1}}
\newcommand{\percy}[1]{\textcolor{red}{\textbf{Percy:} #1}}

\newtheoremstyle{compact}
  {2pt}          % space above
  {2pt}          % space below
  {\normalfont}  % body font, same as definition
  {}             % indent amount
  {\bfseries}    % theorem head font
  {.}            % punctuation after theorem head
  {0.5em}        % space after theorem head
  {}             % theorem head spec

\theoremstyle{compact}
\newtheorem{theorem}{Theorem}

\providecommand{\theoremautorefname}{Theorem}
\providecommand{\algorithmautorefname}{Algorithm}
\providecommand{\listingautorefname}{Listing}

\algrenewcommand\algorithmiccomment[1]{\hfill$\triangleright$\ #1}

\newcommand{\refinery}{\texttt{Refinery}}
\newcommand{\clevr}{CLEVR}

\newcommand{\approach}{\texttt{MeMoDecode}}
\newcommand{\naive}{\emph{Unc}}
\newcommand{\syntax}{\emph{Grammar}}

\newcommand{\llm}{LLM}
\newcommand{\seq}[1]{\pmb{#1}}

\newcommand{\aab}[1]{\textcolor{blue}{\textbf{AAB:} #1}}

\newtcolorbox{summarybox}{
  enhanced,
  colback=gray!8,
  colframe=gray!8,
  boxrule=0pt,
  borderline west={4pt}{0pt}{gray!65},
  arc=1pt,
  left=6pt,
  right=6pt,
  top=3pt,
  bottom=3pt,
  boxsep=1pt,
  before skip=2pt,
  after skip=0pt,
}

% \title{Metamodel Guided Decoding-Time Verification for Consistent DSL Generation with LLMs
% }

% \title{Metamodel-Guided LLM Decoding for \\Consistent DSL Generation}

\title{Metamodel-Guided LLM Decoding for \\Consistent DSL Generation}

\author{\IEEEauthorblockN{Boqi Chen\orcidlink{0000-0002-1451-3603}, Aren A. Babikian\orcidlink{0000-0002-8108-0043}, José Antonio Hernández López\orcidlink{0000-0003-2439-2136}}
\thanks{
Boqi Chen is with the University of Ottawa, Ottawa, Canada (email: boqi.chen@uottawa.ca)

Aren A. Babikian is with McGill University, Montreal, Canada (email: aren.babikian@mcgill.ca)

José Antonio Hernández López is with the University of Murcia, Murcia, Spain (email: joseantonio.hernandez6@um.es)
}}

% \markboth{TRANSACTION ON SOFTWARE ENGINEERING, VOL. XXX, NO. XXX, XXX 2026}%
% {Metamodel-Guided LLM Decoding for Consistent DSL Generation}

\maketitle

\begin{abstract}
Large language models (\llm{}s) are increasingly used to generate software engineering artifacts written in domain-specific languages (DSLs). Correct DSL instances usually must comply with not only the concrete syntax but also the abstract syntax defined by a metamodel augmented with semantic well-formedness constraints.
Existing constrained decoding approaches typically guarantee compliance only with concrete syntax or with a limited set of semantic constraints expressed over the generated text, making metamodel-level constraint consistency difficult to enforce during generation. 
We present \approach{}: a decoding framework that enforces a DSL's concrete syntax, abstract syntax, and well-formedness constraints jointly. 
Specifically, whenever the \llm{} generates a complete metamodel concept, \approach{} incrementally maps the current output to a partial model and verifies whether it can still be completed into a consistent instance. We combine \approach{} with a simple retry mechanism that resumes decoding from the most recent completable snapshot whenever an inconsistency is detected. Under a sound backend verifier, \approach{} never performs worse than grammar-constrained decoding on alignment metrics satisfying natural assumptions. Across four DSLs and six open-weight \llm{}s ranging from 0.6B to 70B parameters, \approach{} achieves the highest consistency, improving consistency by 23\% on average over unconstrained decoding and by 8\% on average over grammar-constrained decoding. It also improves benchmark performance while incurring only modest runtime overhead. These results show that metamodel-guided decoding is an effective approach for enforcing DSL consistency during decoding across diverse DSLs and \llm{} scales.
% \ab{one comment about the name that Marsha mentioned: if we are planning a sequence of papers, lets try to have a commonality in the names of our proposed things (e.g., a common prefix, like MeMoDecode, then next paper is MeMoSearch, then MeMoOptimize) },
 
\end{abstract}


\section{Introduction}
\label{sec:introduction}

\noindent\textbf{Context.}
Modern software engineering (SE) relies on a wide range of artifacts, from architecture design and configuration to query languages and flowcharts, many of which are expressed in \emph{domain-specific languages} (DSLs) tailored to a particular task or domain~\cite{fowler2010domain}. As large language models (\llm{}s) are increasingly used to automate SE activities, a growing portion of these artifacts is generated, either fully or partially, by \llm{}s, from natural language specifications~\cite{joel2025survey,liu2024large}. Beyond traditional SE workflows, DSL artifacts also serve as structured interfaces in LLM-based agentic systems, supporting multi-agent communication~\cite{shi2025flowagent,niu2025flow,liu2026synthesizing} and tool use~\cite{wang2024executable,geng2025jsonschemabench,qin2024toolllm}.

% \percy{One sentence about the relevance to agents? e.g. artifacts for tool calling}

In SE, especially in model-driven engineering, DSL instances are often used as machine-processable artifacts for downstream automation, analysis, model transformation, and code generation~\cite{garzon2015umple,mernik2005and,varro2007model}.
% As a result, correctness extends beyond syntactic validity. 
As a result, a DSL is typically characterized by three complementary aspects: its \emph{concrete syntax}, usually defined by a grammar; its \emph{abstract syntax}, captured by a metamodel; and a set of semantic \emph{well-formedness constraints} over that metamodel, with execution semantics defined on top of them~\cite{bergmann2010incremental,richters2002ocl,bettini2016implementing}. These aspects form a hierarchy of correctness. A generated instance may satisfy the grammar while still violating the metamodel or its constraints, preventing meaningful analysis, transformation, or execution. A DSL instance is \emph{consistent} if it conforms to both the metamodel and the well-formedness constraints~\cite{varro2018towards}.
However, \llm{}s provide \textit{no guarantee} that generated DSL instances satisfy these specifications. They may produce outputs that violate the grammar, the metamodel, or the well-formedness constraints, motivating the need for verification techniques~\cite{netz2024grammar,song2026empirical,chen2026projectional}.

A natural way to improve the reliability of generated DSL instances is \textit{post-hoc verification}, where a complete instance is first generated and then checked or repaired~\cite{chen2025accurate,he2025llm,khamsepour2025impact,lamas2026dslxpert}. However, these approaches discover violations only after a full generation. As a result, they may waste substantial decoding effort and require additional repair steps. Moreover, as LLM generation is conditioned on existing prefixes, when the generated instance contains early structural errors, subsequent parts of the output are generated from an already invalid context, potentially hurting the overall quality~\cite{zhang2024hallucinations,huang2024large}.

These limitations motivate decoding-time approaches that prevent invalid outputs before they are completed. \textit{Grammar-constrained decoding} masks, at each decoding step, tokens that would violate the grammar, thereby guaranteeing syntactic validity~\cite{beurer2023prompting,llguidance2024,willard2023efficient,ugare2024syncode,dong2025xgrammar,kellner2024guiding}. However, operating only on concrete syntax, it provides no guarantee that the generated DSL instance conforms to the metamodel or the associated constraints.

Recent approaches further incorporate semantic constraint verification into the decoding process~\cite{nagy2026chopchop,mundler2025type,agrawal2023monitor}. For example, IterGen~\cite{ugare2025itergen} augments grammar-constrained decoding with semantic checks. However, these checks are typically expressed over the generated text rather than the underlying structured model. Consequently, verification remains tied to a particular textual representation. More importantly, it limits the ability to reason about properties that are naturally defined over objects and references in the metamodel, such as multiplicity constraints, referential integrity, and transitive relationships.
As a result, decoding-time verification and guidance of metamodel-level constraint consistency over the underlying structured representation remain underexplored~\cite{chen2026projectional}.

\noindent\textbf{Approach.}
In this paper, we present \approach{} (\texttt{Me}ta\texttt{Mo}del-guided \texttt{Decod}ing), a decoding framework that simultaneously enforces a DSL's concrete syntax, abstract syntax, and semantic well-formedness constraints. The key idea is to lift verification from the generated text to the underlying model representation. As decoding proceeds, \approach{} incrementally maps the generated text to a partial model and, whenever the \llm{} completes a metamodel concept, verifies whether that partial model remains satisfiable with respect to the metamodel and its well-formedness constraints, that is, whether it can still be completed into a consistent instance. If not, 
the violation is irreparable, and 
\approach{} can immediately backtrack and retry during decoding rather than relying on post-hoc regeneration of the entire output.

\noindent\textbf{Contribution.}
Because verification is performed over the model representation rather than the concrete-syntax text, \approach{} can enforce structured metamodel-level properties that are difficult or impossible to express using grammar-based approaches. Furthermore, since constraints are specified directly over the metamodel, the verification process is independent of the concrete syntax of the DSL and can be applied across different textual representations.
% We also show that, under a sound verifier, \approach{} never produces an DSL instance that is less aligned with the ground truth than one generated by grammar-constrained decoding under identical decoding conditions, provided that consistency is a prerequisite for such alignment. %Intuitively, consistency checking instances artifacts that cannot be completed into a consistent instance and therefore cannot lead to a better-aligned final result.
More specifically, this paper makes the following contributions:
\begin{itemize}[leftmargin=10pt]
    \item We present \approach{}, a decoding framework that simultaneously enforces a DSL's concrete syntax, abstract syntax, and semantic well-formedness constraints by verifying whether partially generated outputs can still be completed into consistent instances and retrying otherwise.
    \item We establish a worst-case guarantee: under a sound verifier and natural assumptions on the alignment metric, \approach{} never produces an output that is less aligned than the one generated by grammar-constrained decoding under identical decoding conditions.
    \item We implement \approach{} by combining grammar-constrained decoding~\cite{llguidance2024} with the \texttt{Refinery} verifier~\cite{marussy2024refinery} and a lightweight decoding-time retry mechanism.
    \item We conduct a systematic evaluation across four DSLs and six open-weight \llm{}s ranging from 0.6B to 70B parameters. The results show that \approach{} consistently improves DSL consistency, increases alignment with the ground truth, and incurs only modest runtime overhead.
\end{itemize}
\noindent\textbf{Added value.}
\approach{} demonstrates that reasoning over the underlying model representation is key to ensuring consistency of a DSL instance. By lifting verification to the metamodel level, it enables the use of established partial modeling and constraint-solving techniques within the decoding process, yielding consistency guarantees independent of the textual representation. Since verification is decoupled from the underlying \llm{}, \approach{} can benefit directly from future advances in both \llm{}s and backend verifiers.
\section{Background and Related Work}
\label{sec:background}
\subsubsection{Domain-Specific Languages}
% A domain-specific language (DSL) has limited expressiveness that is tailored to a particular application domain~\cite{fowler2010domain}. 
A domain-specific language (DSL) is tailored to express the concepts of a particular application domain~\cite{fowler2010domain}.
By focusing on a specific domain, DSLs can provide substantially higher ease of use than general-purpose programming languages for the tasks they are designed to support~\cite{mernik2005and}. For instance, \autoref{lst:clevr-query-color} shows a CLEVR program specified using a DSL. A CLEVR program is a structured, machine-executable representation of a natural-language question about an image~\cite{johnson2017clevr}. The listing illustrates a minimal program to answer the question \textit{``What color is the cube?''}. The program begins with a \texttt{scene} node that returns the full set of objects in the scene, which is then passed to \texttt{filter\_shape} to retain only cube-shaped objects. The \texttt{unique} function asserts that exactly one such object exists, promoting the set of objects to a single object. Finally, \texttt{query\_color} is applied to retrieve the color of that object.

% \aab{I am thinking we could have an informal explanation next to the graph repr, to make the mapping easier from graph to concepts. But this may be ovrkill. In that case, i would add to the caption that this is a graph DSL instance. basically making the connection to DSL explicit.}

\begin{listing}[tbh]
\begin{Verbatim}[
    bgcolor=gray!10,
    fontsize=\small,
    framesep=2mm
]
graph LR
    1["scene"] --> 2["filter_shape_cube"]
    2 --> 3["unique"]
    3 --> 4["query_color"]
\end{Verbatim}
\caption{A CLEVR DSL for \textit{``What color is the cube?''} using the DSL definition of \autoref{fig:clevr_example}}
\label{lst:clevr-query-color}
\end{listing}

% When designing a DSL, language engineers typically define three main aspects of the language (see \autoref{fig:clevr_example}): (1) the concrete syntax, (2) the abstract syntax, and (3) a set of well-formedness constraints.

\subsubsection{DSL specifications}
In SE, and especially in model-driven engineering, a DSL is characterized by three aspects: a concrete syntax, an abstract syntax, and a set of semantic well-formedness constraints~\cite{bergmann2010incremental,richters2002ocl,bettini2016implementing} (see \autoref{fig:clevr_example}). Formally, a DSL specification can be defined as
\(
    \mathcal{D} = (G, M, \mathcal{C}),
\)
where \(G\) is a context-free grammar (CFG) defining the concrete syntax, \(M\) is a metamodel defining the abstract syntax, and \(\mathcal{C}\) is a set of semantic well-formedness constraints over \(M\).

% The concrete syntax specifies the notation that end users employ to interact with the language. This notation can be either textual or graphical. In the case of textual DSLs, which are the focus of this paper, language engineers define the concrete syntax through a formal grammar~\cite{xx}. When a valid program is parsed, the parser produces a \textit{concrete syntax tree} (CST) that captures the syntactic structure of the input according to the grammar. For instance, for the CLEVR example, the language engineer could have defined the context-free grammar depicted in \autoref{fig:clevr_example}. 

The concrete syntax specifies the notation through which users interact with DSL instances. This paper focuses on textual DSLs whose concrete syntax is defined by a CFG \(G\). In the \clevr{} case, a language engineer could define the CFG shown in \autoref{fig:clevr_example}. The \emph{non-terminals} of \(G\) define textual forms for different syntactic categories of the language, such as \texttt{edge}, program \texttt{node}, and the \texttt{type} of a node.
% \aab{Question: how do we define a syntactic unit? is every non-terminal a syntactic unit? What if there are multiple levels of non-terminals? Should we instead say these top-level non-terminals (which are direct children of the root node)? but for this we woudl need to check the MM I guess}
We refer to these completed derivations of selected non-terminals as \textit{syntactic units}. A text sequence \(\seq{y}\) is syntactically valid if it belongs to the language generated by \(G\), i.e., \(\seq{y} \in L(G)\).

%For example, SQL is primarily used as a textual language and therefore has an associated formal grammar that defines its syntax.

%\aab{should we be more formal here?}
%\aab{maybe a more detailed running example, like one ofthe case studies we are running?}

The abstract syntax defines a structural representation of DSL instances and is commonly specified through a \textit{metamodel}~\cite{bettini2016implementing}. A metamodel \(M\) defines the types of elements that may appear in an instance model, including classes, attributes, references, containment references, multiplicities, and inheritance relations, often using a metamodeling language such as EMF Ecore~\cite{steinberg2008emf}. The right side of \autoref{fig:clevr_example} shows the metamodel of the \clevr{} DSL.  We use \(I \models M\) to denote that an instance model \(I\) conforms to the metamodel \(M\).

A set of semantic well-formedness constraints $\mathcal{C}$ complements the metamodel by specifying semantic rules that cannot be expressed by the grammar or the metamodel alone~\cite{varro2018towards}. 
% Typical examples include referential integrity, type consistency, uniqueness requirements, and transitive properties such as acyclicity. 
These constraints are commonly specified using a constraint language~\cite{richters2002ocl,marussy2020specification}. Formally, a constraint $c \in \mathcal{C}$ is a predicate over an instance model: $c: \mathcal{I}_M \rightarrow \{\mathrm{true}, \mathrm{false}\}$. $\mathcal{I}_M$ represents the set of instance models conforming to $M$ and the output shows whether the constraint is satisfied. \autoref{fig:clevr_example} shows three out of the  17 constraints (single-root structure, type-correct dataflow, and acyclicity) of the CLEVR DSL.

To connect a text sequence to an instance model, language engineers typically define \textit{transformation rules} from concrete syntax constructs, such as selected grammar non-terminals, to metamodel elements. Given such rules, a DSL instance written in the concrete syntax can be transformed into an instance model, for example, through parsing and model transformation techniques~\cite{muller2006model,zhang2024supporting}. The coloring in \autoref{fig:clevr_example} illustrates an example of such a correspondence between the grammar and the metamodel. A DSL instance is valid only if the resulting instance model $I$ is \textit{consistent} with respect to $(M, \mathcal{C})$, i.e., it conforms to the metamodel and satisfies all well-formedness constraints: $I \models (M, \mathcal{C}) \iff I \models M \land \forall c \in \mathcal{C},\; c(I)=\mathrm{true}$.

\begin{figure}[tb]
    \centering
    \includegraphics[width=\linewidth]{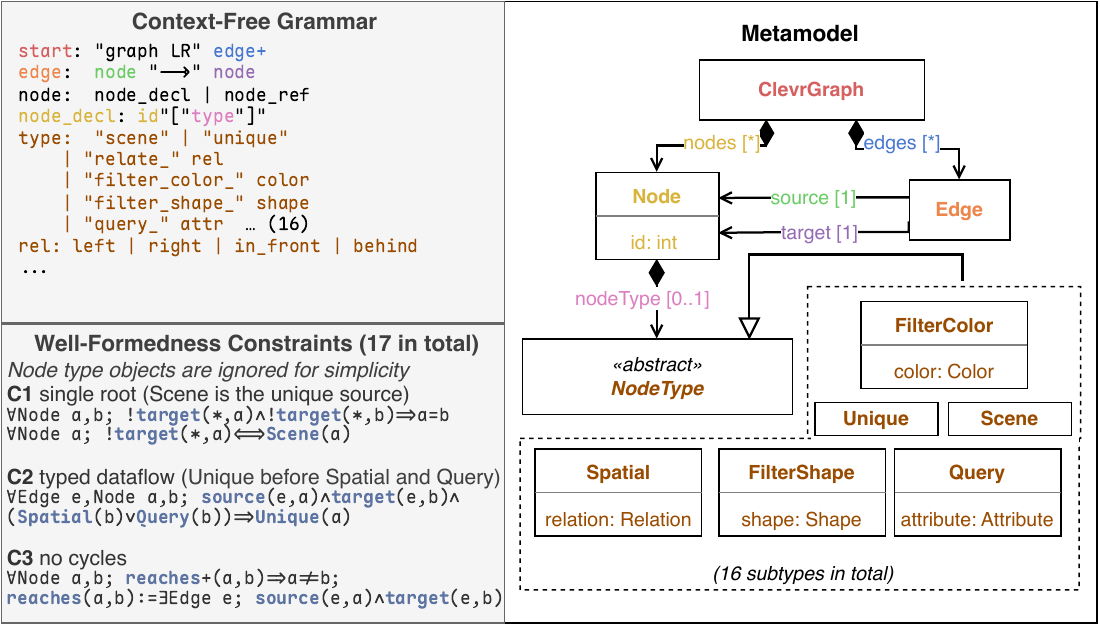}
    \caption{Example DSL for CLEVR programs~\cite{johnson2017clevr}}
    \label{fig:clevr_example}
\end{figure}

\subsubsection{Partial models}
Partial models represent incomplete instances of a metamodel~\cite{famelis2012partial}. Instead of fully specifying all elements, a partial model records what is currently known and leaves the remaining information unspecified~\cite{varro2018towards,semerath2017graph,marussy2020specification}. Following prior work, we view a partial model as an interpretation of metamodel-level \textit{facts}, represented as predicates whose truth values may be known, unknown, or inconsistent. In the \clevr{} DSL, a partial model may already contain facts such as \(\textsc{Node}(n_1)\), \(\textsc{Node}(n_2)\), \(\textsc{Edge}(e_1)\), \(\textsc{source}(e_1,n_1)\), and \(\textsc{target}(e_1,n_2)\), while other facts, such as the node type of \(n_2\) or the remaining outgoing edges of \(n_1\), remain unknown. A complete instance model can be obtained by refining the partial model, i.e., by adding new facts until the objects, references, and attribute values are fully determined. 

\refinery{} is a graph solver for generating and analyzing domain-specific graph models from partial specifications~\cite{marussy2024refinery}. In this paper, we use \refinery{} as the backend solver for reasoning about the partial models constructed during decoding. The design of \approach{} is compatible with any backend that can reason soundly about partial models.

% Its input specification consists of a metamodel, an optional initial partial model, constraints expressed as predicates, and scope information that bounds the generated models. 

\subsubsection{Large language models}
\llm{}s are deep neural networks that typically follow the decoder-only Transformer architecture~\cite{zhao2026survey}. They are trained on massive collections of textual data using the \textit{language modeling task}, which consists of predicting the next token given a sequence of preceding tokens. 
In practice, the generation follows an autoregressive process: an \llm{} receives as input a \textit{prompt} $\seq{x}$ and produces a probability distribution over all tokens in its vocabulary, \(\mathrm{LLM(\seq{x})}=P(\cdot|\seq{x})\). Based on this distribution and a chosen \textit{decoding strategy},
% (e.g., greedy or temperature sampling), 
the next token $t \sim P(\cdot|\seq{x})$ is selected by a \textit{sampler} and appended to the prompt. The updated sequence is fed back into the \llm{}, and the process repeats until a predefined \emph{termination condition} is reached, such as generating an end-of-sequence token, to create an output sequence $\seq{y}$.

\subsubsection{Decoding strategies for LLMs}

%\subsection{Large Language Models}

% \subsubsection{Decoding strategies}

A decoding strategy determines how the next token is selected from the probability distribution \(P(\cdot|\seq{x})\) predicted by an \llm{} from a prompt $\seq{x}$. 
% Existing decoding strategies can be broadly classified into deterministic and stochastic approaches~\cite{shi2024thorough}. 

\noindent
\textbf{Unconstrained decoding}
strategies select tokens from the probability distribution produced by the \llm{} directly, without considering DSL-specific validity constraints. For example, temperature sampling adjusts the sharpness of this distribution by scaling the next-token logits before normalization. Other commonly used strategies include greedy decoding, beam search~\cite{freitag2017beam}, top-$k$ sampling~\cite{fan2018hierarchical}, and top-$p$ sampling~\cite{holtzman2019curious}.

\noindent
\textbf{Constrained decoding.}
In many applications, the generated sequence is required to satisfy constraints beyond those captured by the language model distribution. Constrained decoding refers to decoding methods that restrict token selection using an external constraint specification. Given a set of constraints \(\Phi = \{\phi_1,\ldots,\phi_k\}\), the goal is to generate an output sequence \(\pmb{y}\) such that \(\pmb{y}\) satisfies every constraint in \(\Phi\), while still using the \llm{} probabilities to guide token selection.

When the constraint specification is a CFG \(G\), the constraint is that the final output must belong to the language generated by the grammar, i.e., \(\pmb{y} \in L(G)\). Such methods are commonly referred to as \emph{grammar-constrained decoding}~\cite{beurer2023prompting,llguidance2024,willard2023efficient,ugare2024syncode,dong2025xgrammar,kellner2024guiding}. The core idea is to maintain a parser state induced by the output and, at each decoding step, mask tokens that would make the output impossible to extend into a syntactically valid sequence~\cite{beurer2023prompting,willard2023efficient}. As a result, these approaches enforce concrete-syntax validity during generation, but do not by themselves ensure that the resulting instance model is consistent. 

% to the metamodel or satisfies well-formedness constraints.

% \subsubsection{Grammar-constrained Decoding}

%  However, in many applications it is desirable to enforce additional constraints on the generated text. This leads to the notion of \textit{constrained decoding}, where the token selection process is guided not only by the model probabilities but also by external constraints which aims to enforce some kind of syntactic correctness. More formally, given a set of constraints $\Phi = \{\phi_1,\dots,\phi_C\}$, constrained decoding enforces that the generated output sequence $\pmb{y}\sim \text{LLM}(\pmb{x})$ satisfies every constraint $\phi_i \in \Phi$.

% When $\Phi$ is specified through a formal grammar, such that the generated output $\pmb{y}$ must conform to it, the resulting approaches are referred to as \textit{grammar-constrained decoding}~\cite{beurer2023prompting,llguidance2024,willard2023efficient,ugare2024syncode,dong2025xgrammar}. The core idea behind these methods is to mask invalid tokens at each decoding step, thereby ensuring that the generated sequence $\pmb{y}$ progressively satisfies the constraints imposed by the grammar~\cite{beurer2023prompting,willard2023efficient}.

\noindent
\textbf{Beyond grammar-constrained decoding.}
% Formal grammars are not expressive enough to capture many classes of constraints. For example, they cannot enforce properties such as type consistency or referential integrity. Consequently, although grammar-constrained guarantees syntactic correctness, it does not generally ensure semantic validity of the generated output.
% Since formal grammars are not expressive enough to capture many classes of constraints, 
Several approaches augment grammar-constrained decoding with additional semantic reasoning mechanisms. 
These approaches typically refine the token mask by incorporating complementary information, such as semantic pruners or static analyzers, thereby restricting generation with additional constraints~\cite{nagy2026chopchop,mundler2025type,agrawal2023monitor}. However, they either lack sufficient expressive power to capture complex structured semantic properties~\cite{nagy2026chopchop} or focus on a narrow class of constraints~\cite{agrawal2023monitor,mundler2025type}. For example, ChopChop~\cite{nagy2026chopchop} employs semantic pruners. Given a program space, a semantic pruner returns the subspace of programs that satisfy a particular constraint, such as type consistency. Such pruners cannot fully express all types of constraints. 
% \aab{this is because they are on the program space, and some constraints are i the modeling space, so inexpressible in program space (at least directly). We may want to add this if we have space.}
% For instance, regarding type consistency, they cannot guarantee that all generated outputs are type-correct. 
Other approaches~\cite{mundler2025type,agrawal2023monitor,kapoor2025constrained} focus only on specific language semantic properties rather than general DSL specifications.
% , including type consistency, referential integrity, and API-conformance constraints.

%, which operate over ASTs. \aab{Chopchop does pruning at the PL level, but we lift it to the model level. Percy's Q: are the priuners context-dependent? }\textbf{}

%\percy{I think the discussion on IterGen is important, but currently it is too long}
IterGen~\cite{ugare2025itergen} takes a different approach by building on SynCode~\cite{ugare2024syncode} with an iterative backtracking mechanism that revises previously generated tokens upon constraint violations. It introduces three operations: \texttt{forward} (generate until a syntactic unit is reached), \texttt{backward} (rewind to a previously generated syntactic unit), and \texttt{view} (retrieve the strings associated with a syntactic unit), enabling programmatic inspection and enforcement of semantic constraints during decoding.

However, IterGen exposes intermediate fragments at the concrete-syntax level. In particular, \texttt{view} returns raw strings, so constraint checks must be implemented through string processing or pattern matching. This conflates concrete and abstract syntax, tying constraints to surface details such as delimiters, keywords, and whitespace rather than to the underlying objects, attributes, and references. Users need to re-implement parsing and structural analysis that a proper abstract syntax representation provides directly, making constraints harder to maintain and reason about as the DSL evolves.

\subsubsection{LLMs for DSL Generation}
\llm{}s are increasingly used to automate the creation of software artifacts, many of which are expressed in DSLs~\cite{joel2025survey,liu2024large}. Early work has primarily used unconstrained decoding to study the ability of \llm{}s to generate DSL instances, including software models~\cite{chen2023automated,pan2025llm,ferrari2024model} and executable queries~\cite{lopez2024text2vql,hong2024knowledge}. More recent approaches introduce explicit checking mechanisms, either during decoding through grammar-constrained generation~\cite{park2024grammar,netz2024grammar}, or after generation through post-hoc verification and repair~\cite{chen2025accurate,he2025llm,khamsepour2025impact,lamas2026dslxpert}. 
% Other approaches focus on specific types of languages~\cite{agrawal2023monitor,kapoor2025constrained}. 
These techniques improve the validity of generated DSL instances, but they either focus mainly on concrete syntax or check validity only after a complete artifact is produced.

Existing work typically evaluates the generated artifact with a reference artifact or its expected execution result. We refer to such a measure as an \emph{alignment score} \(s(\seq{y})\), which quantifies how well a generated output \(\seq{y}\) aligns with the intended output. Depending on the DSL, \(s\) can be instantiated as a structural similarity metric, such as an  F1 score over generated elements, or as an execution-based metric, such as result accuracy.

\subsubsection{Positioning} 
% In this paper, we propose a decoding framework that guarantees that the outputs generated by a LLM conform to the concrete syntax and abstract syntax of a domain-specific language while simultaneously satisfying its well-formedness constraints. 

%More specifically, the framework takes as input: (1) a context-free grammar defining the concrete syntax of the language, (2) a meta-model specifying its abstract syntax, (3) a mapping between grammar productions and meta-model elements, \aab{question: is this mapping important for our proposed approach?} and (4) a set of well-formedness constraints defined on top of the meta-model.

\approach{} complements existing \llm{}-based DSL generation approaches by augmenting \llm{} decoding with metamodel-level verification. Unlike unconstrained approaches, which may introduce syntactic or semantic violations during generation, and post-hoc repair approaches, which detect such violations only after a complete artifact has been produced, \approach{} rejects inconsistent structural decisions before they become part of the accepted generation context. Compared with constrained generation approaches, \approach{} is closest to IterGen, but lifts the corresponding checks from concrete-syntax fragments to the metamodel level. As a result, \approach{} performs semantic pruning during decoding while keeping verification independent of a particular textual representation.

\section{Approach}
\label{sec:approach}

\subsubsection{Problem definition}
% Given a prompt \(x\) and a DSL specification \(\mathcal{L}=(G,M,\mathcal{C})\), 
% \approach{} aims to guide the decoding process of an \llm{} to generate a DSL instance \(\seq{y}\) that conforms to the grammar \(G\) and whose induced instance model is consistent with respect to the metamodel \(M\) and constraints \(\mathcal{C}\). Formally, the goal is to generate an output \(\seq{y} \in \mathcal{L}(G)\) such that its induced instance model \(I_{\seq{y}}\) satisfies \(I_{\seq{y}} \models (M,\mathcal{C})\). Moreover, by ensuring these consistency properties, \approach{} also targets to improve the alignment score $s(\seq{y})$ as a side effect, provided that $s(\cdot)$ satisfies certain properties (see \autoref{thm:worst_case}).
Given a natural language prompt \(x\) and a DSL specification \(\mathcal{D}=(G,M,\mathcal{C})\), \approach{} guides the generation process of an \llm{} toward an output \(\seq{y}\) that is valid with respect to both the concrete and abstract definitions. Formally, the goal is to generate a syntactically valid DSL instance \(\seq{y} \in L(G)\) whose induced instance model \(I_{\seq{y}}\) is consistent with the metamodel \(M\) and well-formedness constraints \(\mathcal{C}\): \(I_{\seq{y}} \models (M,\mathcal{C})\). By enforcing these consistency properties during decoding, \approach{} also aims to improve the alignment score \(s(\seq{y})\) as a secondary effect.
% , under the assumptions stated in \autoref{thm:worst_case}.

% As shown in \autoref{fig:architecture}.b, unconstrained decoding may produce  outputs that cannot be parsed at all (\autoref{fig:architecture}b.i).  Grammar-constrained decoding rules out such concrete-syntax errors, but may  still admit outputs whose induced models violate the metamodel  (\autoref{fig:architecture}b.ii) or the well-formedness constraints  (\autoref{fig:architecture}b.iii). The goal of \approach{} is therefore to  enforce all three levels of correctness during decoding: concrete syntax,  metamodel conformance, and well-formedness constraints (\autoref{fig:architecture}b.iv).

\begin{figure*}[tb]
    \centering
    \includegraphics[width=0.9\linewidth]{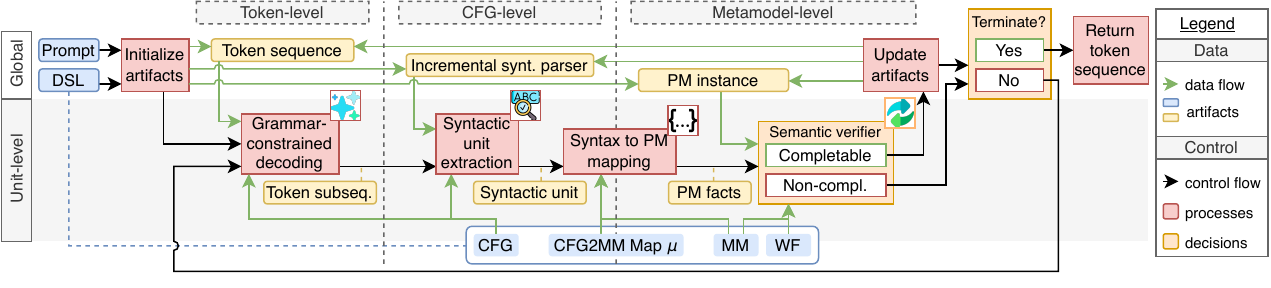}
    \caption{Overview of \approach{} decoding process; MM: metamodel, WF: well-formedness constraints; PM: partial model
    % \percy{Actually add another figure for mapping? Merging figure iii and iv, then we will have some space for the mapping}
    }
    \label{fig:architecture}
\end{figure*}

\begin{figure}[tb]
    \centering
    \includegraphics[width=0.95\linewidth]{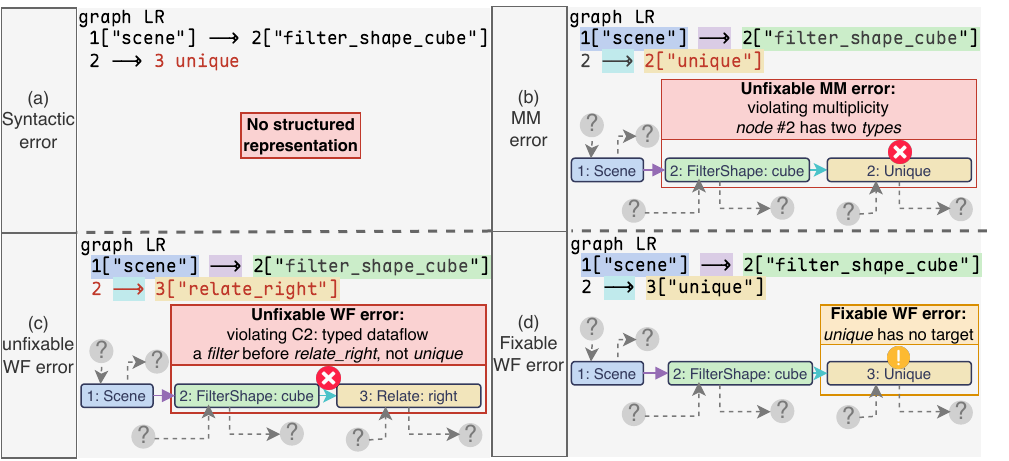}
    \caption{Different partial models, their verification results (boxed), and color-coded correspondence with syntactic units }
    \label{fig:approach_example}
\end{figure}
%%%%%%%%%%%%%%%%%%%%%%%%%%%%%%%%%%%%%%%%%%%%%%%%%%%
\subsubsection{Overview}
\autoref{fig:architecture} summarizes the \approach{} framework. Decoding begins by initializing the artifacts maintained by \approach{}: an empty generated token sequence, a syntactic parser in its initial grammar state, and an empty partial model.
\approach{} maintains these artifacts during generation.
% \approach{} maintains the generated token sequence, an incremental syntactic parser, and the partial model induced by the accepted prefix. 
During decoding, a grammar-constrained decoding approach ensures token-level conformance to the concrete syntax. The generated tokens are then appended to the sequence and consumed by the \emph{syntactic unit extraction} module. Whenever a completed syntactic unit is generated, the \textit{syntax to partial model mapping} maps this unit to metamodel-level facts over \(M\). These facts are tentatively added to a partial model, and a backend \textit{semantic verifier} checks whether the partial model remains completable with respect to \((M,\mathcal{C})\). The artifacts are updated with the unit only if this check succeeds; otherwise, it is rejected and the decoding retries. This process continues until the termination condition is reached.

\subsubsection{Grammar-constrained decoding}
At each decoding step, \approach{} uses a \emph{grammar-guided decoding} module, implemented with \textsc{LLGuidance}~\cite{llguidance2024}, to ensure that the generated sequence remains syntactically valid by assigning zero probability to tokens leading to syntactic violations. 
Such an example for the \clevr{} DSL is shown in \autoref{fig:approach_example}.
% % For example, in the \clevr{} DSL, a token that would lead to the syntactically invalid sequence shown in \autoref{fig:approach_example}.a is assigned zero probability and therefore cannot be generated.

%%%%%%%%%%%%%%%%%%%%%%%%%%%%%%%%%%%%%%%%%%%%%%%%%%%
\subsubsection{Syntactic unit extraction}
\label{subsec:syntactic}

% \approach{} first ensures the output sequence is syntactically valid. At each decoding step, a syntactic parser for grammar-guided decoding, implemented using \textsc{LLGuidance}~\cite{llguidance2024}, maintains the parser state induced by the current sequence and masks tokens that would make the sequence syntactically invalid.
As tokens are generated, the \emph{syntactic unit extraction} module uses an incremental parser to detect when the output completes a selected syntactic unit. A syntactic unit is a completed derivation of a configured grammar non-terminal whose text encodes information relevant to the metamodel. These non-terminals are DSL-specific and are configured by the language engineer. Whenever such a unit is completed, the parser extracts the corresponding grammar fragment and passes it to the mapping module. For example, the highlighted text fragments in \autoref{fig:approach_example}.b, c, and d correspond to information such as node and edge declarations.

%%%%%%%%%%%%%%%%%%%%%%%%%%%%%%%%%%%%%%%%%%%%%%%%%%%

\subsubsection{Mapping syntax to partial models}
\label{subsec:mapping}

\approach{} aims to verify generation at the metamodel level rather than directly over the concrete syntax. To bridge these two levels, \approach{} relies on a mapping \(\mu\) from completed syntactic units to metamodel-level facts. In this paper, we focus on mapping units that directly correspond to either an instance of a metamodel class or a reference between metamodel classes.

In \autoref{fig:approach_example}, the highlight colors indicate which metamodel elements each syntactic unit induces facts for. Given a completed syntactic unit \(u\), \(\mu(u)\) returns the set of facts induced by  \(u\). These facts may introduce objects, assign attribute values, or create references between objects. For example, as shown in \autoref{fig:approach_example}.b, a node declaration \texttt{1["scene"]} is mapped to facts $\textsc{Node}(n_1)$, $\textsc{Scene}(t_1)$, $\textsc{nodeType}(n_1,t_1)$ and $\textsc{id}(n_1,1)$. Similarly, \texttt{2["filter\_shape\_cube"]} introduces facts on a filter and its shape parameter: $\textsc{Node}(n_2)$, $\textsc{FilterShape}(t_2)$, $\textsc{nodeType}(n_2,t_2)$, $\textsc{shape}(t_2,\textsc{Cube})$, and $\textsc{id}(n_2,2)$. Edge declarations induce reference facts. For instance, the completed edge  \texttt{1["scene"]} $\rightarrow$ \texttt{2["filter\_shape\_cube"]} maps to an edge object whose source is the node with ID \(1\) and whose target is the node with ID \(2\): $\textsc{Edge}(e_1)$, $\textsc{source}(e_1,n_1)$, $\textsc{target}(e_1,n_2)$.

During decoding, \approach{} accumulates the facts produced by \(\mu\) into a partial model. The partial model aggregates metamodel-level facts while leaving unspecified elements open for future refinement.
For example, \autoref{fig:approach_example}.b-d visualize the structured representation of the partial model as a graph after successive facts have been aggregated, with grey question marks denoting unknown elements and connections.
% . If a syntactic unit refers to an object that has already been introduced, \(\mu\) reuses the corresponding partial-model object. Otherwise,
% the treatment depends on the DSL specification. 

% For DSLs that allow forward references, \(\mu\) may introduce a partial object whose remaining facts are supplied later. If referential integrity is enforced immediately, a syntactic unit whose reference cannot be resolved is rejected by the verifier.
% yields an incompletable candidate partial model and is rejected by the verifier.

The mapping \(\mu\) can be obtained in different ways depending on the language engineering framework. In \texttt{Xtext}~\cite{xtext}, the relationship between grammar rules and Ecore classes is made explicit through grammar actions and assignments, which makes it possible to derive the mapping automatically. In other settings, \(\mu\) is defined, for example, using a  \texttt{Lark} transformer~\cite{lark_transformer}, or derived from framework-specific annotations.  \approach{} is independent of how \(\mu\) is obtained and supports both manually specified and automatically generated mappings. Once facts are added to the partial model, verification is performed over metamodel facts rather than over surface tokens, enabling \approach{} to enforce structured semantic constraints such as multiplicities and acyclicity.

%%%%%%%%%%%%%%%%%%%%%%%%%%%%%%%%%%%%%%%%%%%%%%%%%%%

\subsubsection{Semantic verification and retry}
After metamodel-level facts are created, \approach{}  verifies the resulting candidate at the level of the partial model. Let \(P\)  denote the partial model induced by the accepted prefix, and let \(\mu(u)\) denote the facts induced by a newly completed syntactic unit \(u\). Adding  \(u\) yields a new candidate partial model \(P' = P \cup \mu(u)\). 

% Verification checks whether \(P'\) can still lead to a consistent instance model.

% still be refined into a complete instance model consistent with the metamodel and semantic well-formedness constraints. 
The verification determines whether facts in the partial model $P'$, either individually or in combination, are inconsistent with the metamodel and well-formedness constraints. Since some violations may still be fixed by subsequent tokens, the verifier focuses on identifying facts that make the partial model impossible to complete into a valid instance.

We formalize this check using the notion of refinement. A refinement operation adds information to a partial model by resolving previously unknown facts, such as introducing missing objects, assigning attribute values, or creating references. We write \(P \sqsubseteq I\) when a complete instance model \(I\) refines a partial model \(P\), i.e., when \(I\) preserves all facts already fixed in \(P\) and resolves the remaining unknowns. A partial model \(P\) is \emph{completable} with respect to \((M,\mathcal{C})\) if there exists at least one complete consistent refinement: $\exists I \: P \sqsubseteq I \land I \models (M,C)$. Accordingly, a violation is \emph{unfixable} if the newly added facts make the candidate partial model not completable. In that case, no continuation of the current prefix can produce a consistent instance without revising previously accepted decisions.

% \percy{briefly mention how we backtrack in the LLM space}

\approach{} uses \refinery{} as a \textit{sound} backend verifier ~\cite{marussy2024refinery} to determine whether the current partial model is completable. If the candidate partial model  \(P'\) is completable, \approach{} commits the syntactic unit \(u\), updates the accepted partial model to \(P'\), and stores a snapshot of the current decoding state. This snapshot includes the accepted tokens, the parser state, and the partial model, so that decoding can later resume from this point if needed.

If \(P'\) is not completable, \approach{} retries with a simple strategy. The decoder rejects both the unit from the token sequence and its induced facts from the partial model. 
Decoding continues from the most recent completable snapshot using temperature sampling. Since the sampling is stochastic, the regenerated sequence may create a different syntactic unit. This retry strategy is intentionally lightweight, allowing us to focus mainly on the effect of metamodel-level verification.

\autoref{fig:approach_example}.b and c show two partial models rejected by the semantic verifier because they are no longer completable. The two cases violate multiplicity and typed dataflow constraints, respectively, and these violations are unfixable. In contrast, the partial model in \autoref{fig:approach_example}.d remains acceptable. While it is incomplete, given that \texttt{Unique} does not yet have a target node, this incompleteness can be fixed by extending the output with more tokens. This verification step ensures that unfixable violations are rejected before they are committed, guiding generation toward a valid DSL instance shown in \autoref{lst:clevr-query-color}.

\subsubsection{Termination conditions}

% The decoding process terminates when the generated output is complete or when some configured budgets are exhausted.
The decoding process terminates when either the LLM produces an end-of-sequence token or the parser state indicates 
% a complete DSL instance is generated, i.e.,
no further tokens are required to complete the DSL instance. 
As a safeguard against non-termination, decoding also stops once the output reaches a maximum number of tokens. Finally, if semantic verification repeatedly rejects candidates and the retry budget is exhausted, \approach{} falls back to using only the grammar-guided decoding module from the most recent completable snapshot. This fallback ensures that decoding can still terminate with a syntactically valid output when the lightweight retry strategy cannot find a completable continuation within the budget.

\subsubsection{Worst-case alignment score guarantee}
The semantic verification rejects only unfixable violations. 
% That is, a completed syntactic unit is rejected only when adding its induced facts makes the partial model incompletable. 
Thus, if the grammar-constrained decoding module eventually produces a consistent instance, \approach{} follows the same generation path and performs no rejection. If a rejection occurs, then the corresponding grammar-based continuation is not completable without revising previously generated tokens. This property yields the following worst-case guarantee.

\begin{theorem}[Worst-case alignment score of \approach{}]
\label{thm:worst_case}
Consider grammar-guided decoding and \approach{} operating under identical decoding conditions, i.e., using the same prompt, \llm{}, sampler, and random seed. Let \(\mathcal{D}=(G,M,\mathcal{C})\) be the DSL specification, and let \(s\) be an alignment scorer that maps each generated output \(\seq{y}\) to a numerical score. Assume that \(s\) satisfies the following properties: \emph{(non-negativity)} $s(\seq{y}) \ge 0$; and \emph{ (precondition)} $s(\seq{y}) = 0$ whenever the instance model induced from $\seq{y}$ is inconsistent with $(M, \mathcal{C})$.

% \begin{itemize}[leftmargin=10pt]
%   \item \emph{(non-negativity)} $s(\seq{y}) \ge 0$; and
%   \item \emph{(precondition)} $s(\seq{y}) = 0$ whenever $\seq{p}$ violates $(M, \mathcal{C})$
%   %;
%    %     equivalently, $\{\, y : s(y) > 0 \,\} \subseteq \{\, y : y \models (M, C) \,\}$.
% \end{itemize}
% \aab{took me a while to check this. maybe add some intuition here?}
Assume that the backend verifier used by \approach{} is sound, i.e., it rejects a completed syntactic unit only when adding the corresponding facts makes the partial model not completable with respect to \((M,\mathcal{C})\). Let \(\seq{y}_{\text{g}}\) and \(\seq{y}_{\text{memo}}\) denote the outputs generated by a grammar-constrained decoding approach and by \approach{}, respectively. Then, for every input, under the previously stated assumptions,
% and any retry strategy,
\(\seq{y}_{\text{memo}}\) achieves an equal or better alignment score than \(\seq{y}_{\text{g}}\) by ensuring semantic consistency on top of syntactic validity.
Formally, 
\(
  s(\seq{y}_{\text{memo}}) \ge s(\seq{y}_{\text{g}}),
\)
with equality whenever no rejection fires.

\end{theorem}
% \aab{what does the value of s(p) mean? for instance the severity of violation like closer to 0 means closer to being violated? We might need to formally define this, and also define what a unfixable violation is.}
%\percy{semantic score is not defined, maybe introduce in the background as part of LLM for DSL generation}
% \begin{adjustwidth}{-2em}{0pt}
% \begin{IEEEproof}
\noindent
\textit{Proof.}
There are two possible cases: (1) $\seq{y}_{\text{g}}$ does not violate $(M,\mathcal{C})$. Since the verification is sound, it follows that $\seq{y}_{\text{g}} = \seq{y}_{\text{memo}}$. Assume, for the sake of contradiction, that $\seq{y}_{\text{g}} \neq \seq{y}_{\text{memo}}$. Then, during the execution of \approach{}, at least one element must have been rejected. This implies that the corresponding partial solution could not be extended into a valid output satisfying $(M,\mathcal{C})$. Since $\seq{y}_{\text{g}}$ itself satisfies $(M,\mathcal{C})$, such a rejection cannot occur, yielding a contradiction; thus, $\seq{y}_{\text{g}} = \seq{y}_{\text{memo}}$. (2) $\seq{y}_{\text{g}}$ violates $(M,\mathcal{C})$. By the precondition and the non-negativity property of s,
    \(
        s(\seq{y}_{\text{g}}) = 0 \leq s(\seq{y}_{\text{memo}}).
    \)

% \begin{enumerate}[leftmargin=30pt]
%     \item

%     \item 
% \end{enumerate}
% \end{IEEEproof}
% \end{adjustwidth}

% \subsection{Implementation details}
% In the current version of \approach{}, we instantiate the components in \autoref{alg:decode} as follows. The semantic verifier $\mathcal{V}$ is realized using \texttt{Refinery}~\cite{marussy2024refinery}, which determines whether the current partial model can still be completed into an instance satisfying $(M,C)$. Each semantic unit is translated into its \texttt{Refinery} representation and added to the partial model before verification. We implement a simple retry strategy $\mathcal{R}$ which discards the rejected unit and resumes decoding from the immediately preceding completable snapshot using temperature sampling. Because sampling is stochastic, the regenerated continuation typically differs from the rejected one. The stop condition $\mathcal{C}$ terminates decoding when the grammar signals a complete program, the output reaches the token limit, or the retry budget is exhausted.

\section{Evaluation}
\label{sec:evaluation}

We evaluate the impact of \approach{} on syntactic validity, consistency, token efficiency, alignment, and runtime. Specifically, we investigate the following research questions:

\begin{itemize}[leftmargin=10pt]
    \item \textbf{RQ1}: How does \approach{} influence the syntactic validity and consistency of generated DSL instances with respect to the metamodel and well-formedness constraints?
    \item \textbf{RQ2}: How token-efficient is \approach{} in producing consistent instances compared with grammar-guided and unconstrained decoding under post-hoc verification?
    \item \textbf{RQ3}: How does \approach{} affect the alignment scores of the generated instances?
    \item \textbf{RQ4}: What runtime overhead does \approach{} introduce compared with other decoding components?
\end{itemize}

\subsection{Setup}
% We introduce the experimental setup used including the benchmarks, \llm{}s, baseline approaches, and metrics. 

\paragraph{Benchmarks}

The experiments are conducted on four popular DSL generation benchmarks grouped into two categories, as summarized in \autoref{tab:benchmarks}. The first category consists of \emph{non-executable DSLs}, which represent structured artifacts such as graph models without execution semantics. In SE, these DSLs are commonly used to describe system architectures, configurations, and taxonomies.

\begin{table}[tb]
\centering
\caption{Summary of evaluation benchmarks.}
\label{tab:benchmarks}
\scalebox{0.85}{
\begin{tabular}{@{}cccp{0.5\linewidth}@{}}
\toprule
Benchmark & Samples & Constraints & Example well-formedness constraints \\
\midrule
\multicolumn{4}{@{}l}{\emph{Non-executable (structural)}} \\
Taxonomy & 100 & 6 & one root; acyclic; $\le\!1$ parent/term\\
Flowchart & 326 & 5 & reachability; decisions have $\ge\!2$ guarded out-edges; no self-loops \\
\addlinespace[2pt]
\multicolumn{4}{@{}l}{\emph{Executable (operational)}} \\
CLEVR & 300 & 17 & one source and one sink; acyclic; type consistency; referential integrity \\
Block & 399 & 11 & value type matches property, e.g., height$\to$integer; type consistency \\
\bottomrule
\end{tabular}
}
\end{table}

\noindent
\textbf{Taxonomy}. We use the WordNet taxonomy construction benchmark~\cite{miller1995wordnet}, which constructs a hypernym graph from a given set of concepts and has been widely used to evaluate LLM-based taxonomy generation~\cite{chen2023prompting,zeng2024chain}. Following prior work, we randomly sample 100 taxonomies, each containing between 11 and 50 terms.

\noindent
\textbf{Flowchart}. We use the PAGED benchmark~\cite{du2024paged}, containing procedure graphs from natural language descriptions. We retain only samples that can be represented as flowcharts and exclude trivial graphs without branches or forks~\cite{chen2025accurate}. 
% The resulting benchmark contains 326 samples.

The second category consists of \emph{executable DSLs}, where each generated program can be executed to produce an output. In SE, this type of DSLs is used to represent executable queries or data-processing pipeline configurations.

\noindent
\textbf{CLEVR}. The CLEVR benchmark ~\cite{johnson2017clevr} generates executable programs over the CLEVR DSL introduced in \autoref{fig:clevr_example} from natural language questions. We use a subset evenly divided among query, count, and judge questions~\cite{chen2025accurate}.

\noindent
\textbf{Block}. We use the test split of the Overnight Block dataset~\cite{wang2015building}, where natural language descriptions of colored blocks are translated into logical forms expressed in a closed lambda function with 29 prefix operators. Each program executes over a simulated block world to produce an end state. 
% The benchmark has also been adopted for LLM-based DSL generation~\cite{wang2023grammar}.
% The Clevr benchmark~\cite{johnson2017clevr} is originally introduced for benchmarking visual question answering tasks but recently also used for LLM-based generation tasks~\cite{chen2025accurate}.

\paragraph{Baselines}
We compare \approach{} with two baseline decoding approaches. To control potential variability, we use the same \llm{}, prompt, sampler, and random seed, differing only in the constraints enforced during decoding. 
% Consequently, any difference in the generated outputs should be be attributed to the applied constraints.

\noindent
\textbf{Unconstrained decoding} (\naive{})  performs standard autoregressive sampling without any constraints, measuring how often the \llm{} produces valid DSL instances by itself. 

\noindent
\textbf{Grammar-constrained decoding} (\syntax{}) masks, at each decoding step, tokens that would violate the DSL context-free grammar. This ensures that the generated output is syntactically valid with respect to the concrete syntax, following prior grammar-constrained decoding approaches~\cite{beurer2023prompting,llguidance2024,willard2023efficient,ugare2024syncode,dong2025xgrammar}.

% while not enforcing metamodel or well-formedness constraints~\cite{beurer2023prompting,willard2023efficient}. This corresponds to the grammar-constrained decoding paradigm adopted by prior work~\cite{beurer2023prompting,llguidance2024,willard2023efficient,ugare2024syncode,dong2025xgrammar}. 

% \noindent
% \textbf{\approach{}} further augments syntax-guided decoding with metamodel-based verification paired with a simple retry mechanism. Whenever a semantic unit is completed, it validates the resulting partial model against $(M, C)$. If a non-repairable violation is detected, decoding backtracks and retry using the temperatured sampling.

\paragraph{LLMs}
We evaluate six open-weight, instruction-tuned \llm{}s from three families: Qwen3 (0.6B, 8B, and 32B parameters)~\cite{yang2025qwen3}, Gemma-4 (E2B and 26B)~\cite{gemma4}, and Llama3.3-70B~\cite{llama33}. These \llm{}s span different capacities, allowing us to evaluate whether the effects of \approach{} are limited to smaller \llm{}s that frequently generate inconsistent outputs or persist for stronger \llm{}s that already exhibit high capability. We additionally report \texttt{GPT-5.5}~\cite{gpt55} under unconstrained decoding as a state-of-the-art reference.

% The constrained configurations are evaluated only on open-weight \llm{}s because \approach{} intervenes directly in the decoding process. It requires access to next-token logits, token-level masking, and the ability to discard generated tokens when a semantic unit is rejected. Current proprietary APIs expose only sampled outputs. Although the \texttt{GPT-5.5} API supports constrained decoding, it is limited to JSON schema-based constraints, which are insufficient to implement either grammar-constrained decoding or \approach{}. As a result, \texttt{GPT-5.5} is evaluated only under unconstrained decoding. 
% \approach{} is \llm{}-agnostic. It interacts with the underlying \llm{} only through the standard decoding interface of next-token logits and token masking. Consequently, any \llm{} exposing this interface, including proprietary \llm{}s if such access becomes available, can be used without modification.

Although \approach{} is \llm{}-agnostic, we evaluate it only on open-weight \llm{}s because it requires direct access to the decoding process, including next-token logits, token-level masking, and the ability to reject generated tokens. Current proprietary APIs typically expose only sampled outputs or limited constrained decoding interfaces. For example, while the \texttt{GPT-5.5} API supports constrained decoding, its constraints are limited to JSON schemas and therefore cannot express either grammar-constrained decoding or the partial-model verification performed by \approach{}. Consequently, \texttt{GPT-5.5} is evaluated only under unconstrained decoding. 

\paragraph{Metrics}
For RQ1, we report two metrics that measure the correctness of the generated DSL instances. \textbf{Syntactic validity} (\emph{Syn}) is the fraction of outputs that parse according to the DSL's concrete syntax. \textbf{Consistency} (\emph{Con}) is the fraction of outputs that are additionally consistent with the metamodel and its well-formedness constraints $(M, \mathcal{C})$. Note that every consistent instance must also be syntactically valid.

For RQ3, we evaluate how closely the generated instance aligns with the ground truth.  For the structural DSLs (Taxonomy and Flowchart), we compute an element-wise \textbf{Soft F1}~\cite{franti2023soft} by matching nodes and relations between the generated and ground truth instances, regardless of validity. Since it assigns partial credit to inconsistent outputs, it does not satisfy the \textit{precondition} assumptions of \autoref{thm:worst_case} and is reported only as a reference metric. To check the theorem empirically on structural DSLs, we also report \textbf{Hard-0 F1$^*$}, which assigns a score of $0$ to inconsistent instances and otherwise equals the corresponding Soft F1 score, forcing the \textit{precondition} property.

For the executable DSLs, we report \textbf{execution accuracy}: \emph{denotation accuracy} (Block), measuring whether the generated program evaluates to the intended state, and \emph{answer accuracy} (CLEVR), measuring whether executing the program produces the correct answer. Execution accuracy satisfies both \textit{non-negativity} and \textit{precondition} properties of \autoref{thm:worst_case} naturally. 

For the efficiency metrics (RQ2 and RQ4), we compute the average number of generated tokens per instance (RQ2) and the average runtime per sample, broken down by the different components of \approach{} (RQ4).

\paragraph{Experiment configurations}
Unless otherwise stated, each experiment is repeated with three random seeds, and we report the average performance. A sampling temperature of 0.7 is used to enable effective retry. All experiments are performed with few-shot prompting. To ensure a consistent evaluation across \llm{}s, we disable the reasoning (thinking) mode and evaluate only the generation capability. We set the maximum number of tokens to be generated to 5096, and the maximum retry attempts to 50. All experiments with local \llm{}s are conducted with a single NVIDIA RTX Pro 6000 GPU, an Intel Xeon Platinum 8470Q CPU, and 110 GB of system memory. 
% Additional experimental implementation and configuration are provided in the accompanying artifact~\cite{memodecode}.

\subsection{RQ1: Impact on consistency}
\noindent
\textbf{Rationale and setup.}
In RQ1, we evaluate the effect of decoding-time verification during generation. Specifically, we investigate how \approach{} affects syntactic validity and consistency with respect to $(M, \mathcal{C})$ compared with the \naive{} and \syntax{} baselines, and whether these improvements generalize across different \llm{} families and \llm{} scales.

\begin{table*}[tb]
\centering
\caption{Syntactic validity (\textbf{Syn}) and consistency (\textbf{Con}), in \%, averaged over three runs, across the four DSL benchmarks. The best consistency per benchmark and \llm{} is in \textbf{bold}. GPT-5.5 (unconstrained only) is shown as a reference.}
\label{tab:rq1-accuracy}
\setlength{\tabcolsep}{3pt}
\resizebox{0.95\textwidth}{!}{%
\begin{tabular}{lcccccccccccccccccccccccc}
\toprule
\llm{} & \multicolumn{6}{c}{Taxonomy} & \multicolumn{6}{c}{Flowchart} & \multicolumn{6}{c}{Block} & \multicolumn{6}{c}{CLEVR} \\
\cmidrule(lr){2-7} \cmidrule(lr){8-13} \cmidrule(lr){14-19} \cmidrule(lr){20-25}
 & \multicolumn{2}{c}{Unc} & \multicolumn{2}{c}{Grammar} & \multicolumn{2}{c}{\approach{}} & \multicolumn{2}{c}{Unc} & \multicolumn{2}{c}{Grammar} & \multicolumn{2}{c}{\approach{}} & \multicolumn{2}{c}{Unc} & \multicolumn{2}{c}{Grammar} & \multicolumn{2}{c}{\approach{}} & \multicolumn{2}{c}{Unc} & \multicolumn{2}{c}{Grammar} & \multicolumn{2}{c}{\approach{}} \\
\cmidrule(lr){2-3} \cmidrule(lr){4-5} \cmidrule(lr){6-7} \cmidrule(lr){8-9} \cmidrule(lr){10-11} \cmidrule(lr){12-13} \cmidrule(lr){14-15} \cmidrule(lr){16-17} \cmidrule(lr){18-19} \cmidrule(lr){20-21} \cmidrule(lr){22-23} \cmidrule(lr){24-25}
 & Syn & Con & Syn & Con & Syn & Con & Syn & Con & Syn & Con & Syn & Con & Syn & Con & Syn & Con & Syn & Con & Syn & Con & Syn & Con & Syn & Con \\
\midrule
Qwen3-0.6B & 99.7 & 36.0 & 100.0 & 36.0 & 100.0 & \textbf{66.7} & 78.5 & 71.9 & 99.6 & 89.1 & 99.3 & \textbf{98.6} & 51.4 & 49.0 & 100.0 & 93.7 & 100.0 & \textbf{98.8} & 44.3 & 5.3 & 99.9 & 5.6 & 99.9 & \textbf{31.2} \\
Qwen3-8B & 99.3 & 69.0 & 99.3 & 69.0 & 99.3 & \textbf{78.3} & 93.8 & 90.4 & 100.0 & 96.1 & 100.0 & \textbf{96.4} & 75.0 & 74.0 & 100.0 & 94.9 & 99.9 & \textbf{95.7} & 90.8 & 57.1 & 99.9 & 59.8 & 100.0 & \textbf{65.6} \\
Qwen3-32B & 99.7 & 79.3 & 99.7 & 79.3 & 100.0 & \textbf{88.7} & 87.7 & 85.4 & 100.0 & \textbf{97.2} & 100.0 & \textbf{97.2} & 76.8 & 75.6 & 100.0 & 92.6 & 100.0 & \textbf{96.8} & 93.4 & 62.0 & 100.0 & 60.3 & 100.0 & \textbf{67.9} \\
Gemma4-E2B & 64.0 & 22.7 & 92.7 & 29.3 & 90.7 & \textbf{44.3} & 0.0 & 0.0 & 99.7 & 89.0 & 99.7 & \textbf{89.6} & 60.4 & 55.6 & 100.0 & 83.5 & 100.0 & \textbf{88.6} & 68.4 & 19.9 & 99.2 & 20.1 & 99.7 & \textbf{30.8} \\
Gemma4-26B & 98.7 & 71.0 & 99.3 & 70.7 & 99.7 & \textbf{78.3} & 1.1 & 1.1 & 100.0 & 98.0 & 100.0 & \textbf{98.3} & 80.0 & 76.9 & 99.4 & 92.2 & 99.2 & \textbf{96.1} & 97.3 & 61.1 & 99.9 & 61.9 & 99.9 & \textbf{69.6} \\
Llama3.3-70B & 99.3 & 82.3 & 100.0 & 82.0 & 100.0 & \textbf{91.3} & 98.9 & 96.2 & 100.0 & 97.2 & 100.0 & \textbf{97.6} & 78.9 & 76.9 & 100.0 & 94.9 & 100.0 & \textbf{99.4} & 93.1 & 64.8 & 99.8 & 58.2 & 99.9 & \textbf{77.0} \\
\midrule
GPT-5.5 & 100.0 & 98.7 & -- & -- & -- & -- & 100.0 & 100.0 & -- & -- & -- & -- & 99.2 & 95.5 & -- & -- & -- & -- & 100.0 & 82.9 & -- & -- & -- & -- \\
\bottomrule
\end{tabular}}
\end{table*}

\noindent
\textbf{Results.}
\autoref{tab:rq1-accuracy} reports the syntactic validity and consistency of the generated instances.  
\syntax{} achieves near-perfect syntactic validity in almost all cases, and \approach{} largely preserves this property. The largest reduction is only 2\% (92.7\% vs.\ 90.7\% for \texttt{Gemma4-E2B} on Taxonomy), while the remaining differences are within 0.5\%. In contrast, \naive{} decoding is substantially less reliable. For example, \texttt{Gemma4-E2B} and \texttt{Gemma4-26B} output almost no syntactically valid Flowcharts (0.0\% and 1.1\%), showing that smaller \llm{}s cannot follow the syntactic structure on their own.

% Manual inspection shows that most syntactic failures arise when the \llm{} exhausts the token budget due to repetition, causing unparsable outputs.

% Manual inspection shows that most syntactic failures arise when the \llm{} repeatedly generates tokens until the token budget is exhausted, producing truncated and therefore unparseable outputs.

\approach{} achieves the highest consistency for every \llm{} and benchmark combination, outperforming \syntax{} throughout. The only exception is Flowchart with \texttt{Qwen3-32B}, where both approaches achieve 97.2\%. The improvements are often substantial, especially for smaller \llm{}s. For example, \approach{} increases the consistency of \texttt{Qwen3-0.6B} from 36.0\% to 66.7\% on Taxonomy (+30.7\%) and from 5.6\% to 31.2\% on CLEVR (+25.6\%). Grammar constraints alone are therefore insufficient and can even sometimes reduce consistency. On CLEVR, \syntax{} decreases the consistency of \texttt{Llama3.3-70B} from 64.8\% under \naive{} decoding to 58.2\%, whereas \approach{} increases it to 77.0\%, demonstrating the importance of enforcing metamodel and well-formedness constraints in addition to the syntax.

Although the largest improvements are observed for smaller \llm{}s, they persist for larger ones. Compared with \syntax{}, \approach{} improves the consistency of \texttt{Llama3.3-70B} by 18.8\% on CLEVR and by 4.5\% on Block. Even the unconstrained \texttt{GPT-5.5} reference does not achieve perfect consistency, reaching 95.5\% on Block and 82.9\% on CLEVR. This result indicates that metamodel-based verification remains beneficial even for highly capable \llm{}s.

\begin{summarybox}
\textbf{Answer to RQ1:}
\approach{} preserves the near-perfect syntactic validity of \syntax{} (within 2\%) while achieving the highest consistency across all \llm{}s and benchmarks, never performing worse than \syntax{}. 
% Because \syntax{} often leave consistency unchanged and can even reduce it, 
% These results demonstrate that decoding-time metamodel verification is essential for consistent DSL generation. 
Although the largest gains are observed for smaller \llm{}s, improvements persist across all evaluated \llm{}s.
\end{summarybox}

\subsection{RQ2: Token efficiency}

\noindent
\textbf{Rationale and setup.}
The consistency improvements in RQ1 come at additional token costs since retries are performed during decoding. RQ2 checks if this cost is justified by asking how efficiently, in terms of generated tokens, \approach{} produces consistent instances compared to baselines. 

A common alternative to constrained decoding is \emph{post-hoc verification}~\cite{he2025llm,khamsepour2025impact}. In this case, \llm{}s generate multiple outputs, and the verifier retains only consistent instances. Higher consistency may then be achieved by generating more outputs, making token efficiency the quality of interest.

We evaluate this trade-off on \texttt{Qwen3-8B}, a representative mid-sized \llm{}. We use \emph{best-of-$N$ consistency}, defined as the fraction of prompts for which at least one of $N$ sampled instances is consistent. This simulates a post-hoc verifier that always selects a consistent candidate whenever one exists. We reuse the three runs from RQ1 for \approach{} and perform additional runs for \naive{} and \syntax{} until their average cumulative token count matches that of \approach{}, whose token counts include all retries. 
% As $N$ increases, 
We record the best-of-$N$ consistency with the corresponding expected output tokens, yielding a consistency-token trade-off curve for each approach.

\begin{figure}[tb]
    \centering
    \includegraphics[width=0.8\linewidth]{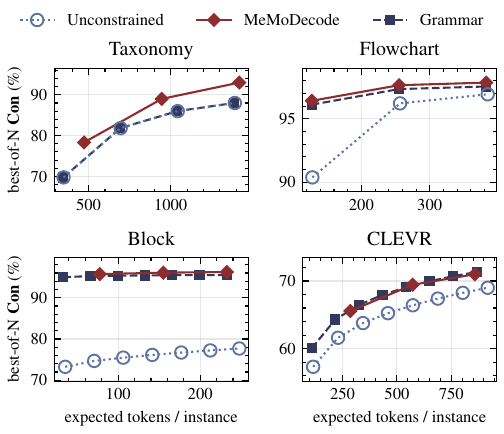}
    \caption{Token efficiency of decoding approaches on Qwen3-8B}
    \label{fig:bestofN}
\end{figure}

\noindent
\textbf{Results.}
The  best-of-$N$ consistency against the expected number of output tokens per instance for \texttt{Qwen3-8B} is plotted in \autoref{fig:bestofN}. As expected, \approach{} generates more output tokens per instance than the baselines. At $N=1$, its curve starts farthest to the right on every benchmark (e.g., approximately 470 versus 350 tokens on Taxonomy), reflecting the additional retries performed during decoding. \naive{} and \syntax{} generate a similar number of tokens because \syntax{} only applies token masking without introducing retries.

Despite this overhead, \approach{} is consistently more token-efficient at an equal token budget. Compared with \naive{}, it achieves substantially higher best-of-$N$ consistency on all four benchmarks,
% , most notably on Block, where \naive{} plateaus near 78\% while \approach{} reaches approximately 96\%. 
indicating that repeatedly sampling unconstrained outputs rarely repairs malformed instances. Compared with the stronger \syntax{} baseline, \approach{} is clearly more efficient on Taxonomy, exceeding it by about 5\% at the same token budget (approximately 93\% versus 88\%) while reaching its best consistency with substantially fewer tokens. It is also slightly more efficient on Flowchart and Block, and comparable on CLEVR, where the two curves nearly overlap because the higher per-generation cost of \approach{} is offset by more samples from \syntax{}.
Moreover, the baseline curves for \naive{} and \syntax{} flatten below \approach{} on benchmarks where it achieves higher consistency.
% , most notably on Taxonomy, where \syntax{} plateaus near 88\% while \approach{} reaches approximately 93\%. 
This result suggests that, under an equal token budget, allocating additional tokens to repeated generation alone cannot always lead to the consistency improvements achieved by \approach{}.

The figure reports only output tokens. Since post-hoc verification requires a new generation pass for every candidate instance, whereas \approach{} performs retries within a single decoding session, accounting for input tokens would likely further favor \approach{}. In practice, this advantage is partially mitigated by prompt caching, although cached input tokens are still billed by most proprietary APIs.

\begin{summarybox}
\textbf{Answer to RQ2:}
Despite its per-generation retry overhead, \approach{} is at least as token-efficient as both baselines and notably more efficient on three of the four benchmarks. These results suggest that, in many cases, allocating tokens to retry during decoding is more effective than allocating them to repeated generation.
\end{summarybox}

\begin{table*}[tbh]
\centering
\caption{Alignment scores (\%) per approach across the four benchmarks, best F1$^*$ or accuracy per benchmark and \llm{} is shown in \textbf{bold}, GPT-5.5 as an unconstrained reference.}
\label{tab:rq2-alignment}
\setlength{\tabcolsep}{3pt}
\resizebox{0.9\textwidth}{!}{%
\begin{tabular}{lcccccccccccccccccc}
\toprule
\llm{} & \multicolumn{6}{c}{Taxonomy} & \multicolumn{6}{c}{Flowchart} & \multicolumn{3}{c}{Block} & \multicolumn{3}{c}{CLEVR} \\
\cmidrule(lr){2-7} \cmidrule(lr){8-13} \cmidrule(lr){14-16} \cmidrule(lr){17-19}
 & \multicolumn{2}{c}{Unc} & \multicolumn{2}{c}{Grammar} & \multicolumn{2}{c}{\approach{}} & \multicolumn{2}{c}{Unc} & \multicolumn{2}{c}{Grammar} & \multicolumn{2}{c}{\approach{}} & \multicolumn{1}{c}{Unc} & \multicolumn{1}{c}{Grammar} & \multicolumn{1}{c}{\approach{}} & \multicolumn{1}{c}{Unc} & \multicolumn{1}{c}{Grammar} & \multicolumn{1}{c}{\approach{}} \\
\cmidrule(lr){2-3} \cmidrule(lr){4-5} \cmidrule(lr){6-7} \cmidrule(lr){8-9} \cmidrule(lr){10-11} \cmidrule(lr){12-13}
 & F1 & F1$^*$ & F1 & F1$^*$ & F1 & F1$^*$ & F1 & F1$^*$ & F1 & F1$^*$ & F1 & F1$^*$ & Acc & Acc & Acc & Acc & Acc & Acc \\
\midrule
Qwen3-0.6B & 20.0 & 8.6 & 20.0 & 8.6 & 19.9 & \textbf{14.4} & 53.9 & 49.6 & 68.1 & 61.3 & 68.1 & \textbf{67.7} & 11.9 & 18.6 & \textbf{19.3} & 1.0 & 1.8 & \textbf{6.1} \\
Qwen3-8B & 53.2 & 36.5 & 53.2 & 36.5 & 53.1 & \textbf{40.7} & 71.5 & 69.0 & 76.2 & 73.4 & 76.2 & \textbf{73.7} & 35.1 & 38.8 & \textbf{39.0} & 44.9 & 46.4 & \textbf{49.8} \\
Qwen3-32B & 59.1 & 46.7 & 59.0 & 46.6 & 59.4 & \textbf{52.3} & 67.5 & 65.8 & 76.9 & \textbf{75.0} & 76.9 & \textbf{75.0} & 42.5 & 45.3 & \textbf{46.2} & 55.0 & 53.8 & \textbf{58.7} \\
Gemma4-E2B & 26.7 & 9.1 & 36.4 & 11.7 & 36.6 & \textbf{18.3} & 0.0 & 0.0 & 72.3 & 66.8 & 72.3 & \textbf{67.3} & 30.7 & 38.9 & \textbf{39.5} & 12.2 & 11.8 & \textbf{15.9} \\
Gemma4-26B & 54.9 & 39.0 & 55.5 & 39.0 & 55.6 & \textbf{43.4} & 0.8 & 0.8 & 77.2 & 75.8 & 77.2 & \textbf{76.0} & 50.6 & 55.6 & \textbf{56.1} & 52.7 & 53.1 & \textbf{59.2} \\
Llama3.3-70B & 60.4 & 50.5 & 60.2 & 49.9 & 60.4 & \textbf{55.5} & 76.8 & 74.9 & 77.7 & 75.7 & 77.6 & \textbf{76.0} & 42.8 & 51.5 & \textbf{52.2} & 58.0 & 52.0 & \textbf{65.7} \\
\midrule
GPT-5.5 & 71.8 & 71.0 & -- & -- & -- & -- & 78.4 & 78.4 & -- & -- & -- & -- & 57.9 & -- & -- & 82.2 & -- & -- \\
\bottomrule
\end{tabular}}
\end{table*}

\subsection{RQ3: Alignment scores}
\noindent
\textbf{Rationale and setup.}
RQ1 and RQ2 focus on the validity of the generated DSL instances. Still, the objective of DSL generation is often to produce instances that correctly capture the intention of the natural language input. Under the assumptions on the alignment scorer, \autoref{thm:worst_case} predicts that \approach{} cannot achieve a lower alignment score than \syntax{}. In this RQ, we evaluate this prediction empirically. We reuse the generation results from RQ1 and evaluate them using the alignment metrics. For the structural DSLs, we report both the standard \emph{Soft F1} score and its \emph{Hard-0 F1$^*$} variant. For the executable DSLs, we report denotation \emph{accuracy} for Block and answer \emph{accuracy} for CLEVR. These metrics instantiate the alignment scorer $s$ in the theorem.

\noindent
\textbf{Results.}
\autoref{tab:rq2-alignment} reports the alignment scores across the four benchmarks. For the structural DSLs, the conventional Soft F1 changes only marginally between \syntax{} and \approach{}, remaining within $\pm0.4\%$.
%  on Taxonomy and showing almost no change on Flowchart.
This result indicates that \approach{} largely preserves the alignment quality of valid outputs. In contrast, Hard-0 F1$^*$ improves for every \llm{} except one tie (\texttt{Qwen3-32B} on Flowchart), with gains of up to 6.6\% (\texttt{Gemma4-E2B} on Taxonomy, 11.7\%\,$\rightarrow$\,18.3\%).

A similar pattern is observed for the executable DSLs, where invalid DSL instances receive an execution accuracy of zero. \approach{} improves execution accuracy for every \llm{} and benchmark, with the largest gain on CLEVR (\texttt{Llama3.3-70B}, 52.0\%\,$\rightarrow$\,65.7\%, +13.7\%). These improvements arise since \approach{} fixes inconsistent instances while preserving the quality of valid ones. However, none of the evaluated open-weight \llm{}s surpasses the unconstrained \texttt{GPT-5.5}, indicating that while decoding-time verification improves alignment, the \llm{}'s capability remains important for the achievable alignment score.

\begin{table}[tb]
\centering
\caption{Win/tie/loss of \approach{} vs. \syntax{} (See \autoref{thm:worst_case}) on the per-sample alignment score, pooled over six \llm{}s and three runs; $p$ is a one-sided Wilcoxon signed-rank test (alternative hypothesis: \approach{} $>$ \syntax{})}
\label{tab:rq2-wtl}
\setlength{\tabcolsep}{4pt}
\begin{tabular}{ccccc}
\toprule
Benchmark & W & T & L & $p$ \\
\midrule
Taxonomy & 234 & 1566 & 0 & $<\!10^{-4}$ \\
Flowchart & 109 & 5759 & 0 & $<\!10^{-4}$ \\
Block & 43 & 7139 & 0 & $<\!10^{-4}$ \\
CLEVR & 329 & 5070 & 1 & $<\!10^{-4}$ \\
\bottomrule
\end{tabular}
\end{table}

To further validate the theorem at the instance level, \autoref{tab:rq2-wtl} reports the per-sample win/tie/loss counts of \approach{} against \syntax{}, aggregated across all \llm{}s and runs. \approach{} records zero losses on three of the four benchmarks with statistically significant improvements ($p<10^{-4}$). The only loss occurs on a single CLEVR instance. Manual inspection shows that this loss is due to an incorrect backend rejection caused by integer overflow, triggered by an excessively large numerical identifier generated by the \llm{}. Although this issue can be fixed, we retain it to show that the theoretical guarantee also assumes a correct implementation.

\begin{summarybox}
\textbf{Answer to RQ3:}
\approach{} consistently improves alignment metrics satisfying the assumptions of \autoref{thm:worst_case}. It loses to \syntax{} on only one out of all evaluated instances, empirically confirming the theorem. The only exception is caused by an incorrect rejection by the backend verifier rather than a limitation of the theorem.
\end{summarybox}

\subsection{RQ4: Runtime analysis}

\noindent
\textbf{Rationale and setup.}
The previous RQs show that \approach{} improves consistency, token efficiency, and alignment. These improvements come at the cost of additional decoding-time verification and retries, which introduce runtime overhead. RQ4 quantifies this overhead by decomposing the average runtime per generated instance.

We report results for three representative open-weight \llm{}s spanning small, medium, and large scales (\texttt{Qwen3-0.6B}, \texttt{Qwen3-8B}, and \texttt{Llama3.3-70B}).
% , covering the range of \llm{} capacities evaluated in this paper. 
We omit \naive{} because its runtime closely matches \syntax{}, which differs only by token masking and performs no retries.

\begin{figure*}[tb]
    \centering
    \includegraphics[width=0.9\linewidth]{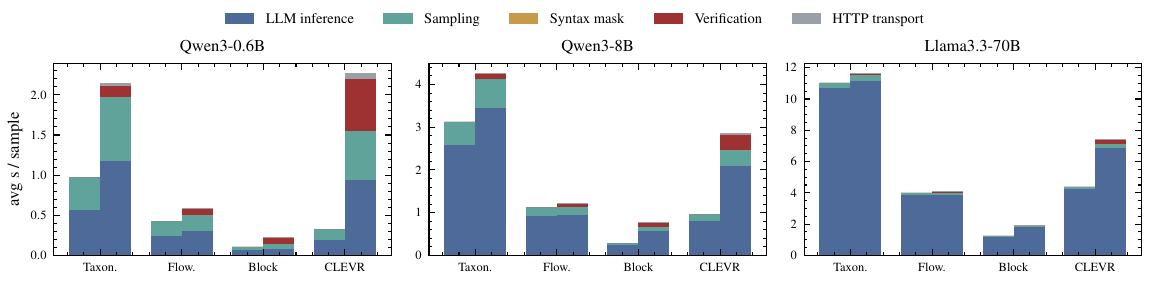}
    \caption{Average runtime breakdown; for each benchmark, the two bars correspond to \syntax{} and \approach{}, respectively}
    \label{fig:runtime}
\end{figure*}

\noindent
\textbf{Results.}
\autoref{fig:runtime} breaks down the average runtime per generated instance into \emph{\llm{} inference, sampling, syntax masking, verification, and HTTP transport} to the backend verifier, comparing \syntax{} and \approach{}. Across all benchmarks and \llm{}s, the verification step contributes only a small fraction of the total runtime. This overhead becomes almost negligible for larger \llm{}s, where runtime is dominated by \llm{} inference. For example, on Taxonomy with \texttt{Llama3.3-70B}, \approach{} increases the runtime only from around 11.2s to around 11.8s per generated instance. 

While \approach{} is slower than \syntax{}, the overhead is almost entirely due to the extra \llm{} inference required for retrying. This effect is most pronounced on CLEVR, whose complex metamodel and well-formedness constraints lead to the largest retry overhead. 
% For \texttt{Qwen3-8B}, the average runtime increases from approximately 1.0\,s to 3.0\,s per generated instance, with nearly all of the increase attributable to additional \llm{} inference. 
This relative overhead is substantially smaller for stronger \llm{}s. On CLEVR, the runtime ratio between \approach{} and \syntax{} decreases from approximately $7\times$ for \texttt{Qwen3-0.6B} to about $3\times$ for \texttt{Qwen3-8B}, and to less than $2\times$ for \texttt{Llama3.3-70B}. 
% Overall, the runtime overhead of \approach{} is dominated by the additional \llm{} inference caused by retries rather than by verification. 
% This observation suggests that future work should focus on optimizing the retrying mechanims
% reducing the need for retries rather than optimizing the verification procedure.

\begin{summarybox}
\textbf{Answer to RQ4.}
The runtime of verification itself is not significant. The overhead of \approach{} is dominated by the additional \llm{} inference required by retries and decreases as \llm{} capability increases, making retry the primary runtime bottleneck.
\end{summarybox}

\subsection{Discussion}
% \percy{practical implications, retrying mechanism, limitations}
% In this subsection, we discusses the implications from the findings in the fours research questions.

\noindent
\textbf{DSL correctness goes beyond syntactic validity.}
RQ1-RQ3 show that syntactic validity alone is insufficient to ensure correctness of generated DSL instances. Grammar-constrained decoding is effective at enforcing the concrete syntax, but it can still generate inconsistent instances and, in some cases, even achieve lower consistency than unconstrained decoding. By verifying the generated outputs at the metamodel level, \approach{} provides an additional layer of correctness beyond syntax. This leads to systematic improvements in consistency and, for alignment metrics that treat invalid instances as unusable, also improves alignment with the intended output.

\noindent
\textbf{LLM capacity and guarantee.}
RQ1 and RQ3 show that larger \llm{}s can improve both consistency and alignment. However, they still fail to produce consistent DSL instances in many cases. Even GPT-5.5 generates inconsistent instances on three of the four benchmarks, reaching only 82.9\% consistency on \clevr{}. These results suggest that while \llm{} scaling reduces the frequency of inconsistent instances, it does not eliminate the need for formal verification, especially when the generated results are intended for downstream SE applications.

\noindent
\textbf{Runtime overhead and retry.}
RQ4 shows that the main runtime overhead of \approach{} comes from the additional \llm{} inference introduced by regenerating tokens after rejected syntactic units, while metamodel-level verification accounts for only a small fraction of the total runtime. 
% This result suggests that the main cost is due to regenerating tokens after rejected syntactic units. 
Therefore, future work should focus on more efficient retry and recovery strategies, such as finer-grained backtracking, adaptive retry budgets, verifier-guided resampling, or decoding-time repair mechanisms that reduce the need for regeneration.

\noindent
\textbf{Limitations and future work.}
While \approach{} improves over both baseline decoding approaches, several limitations remain for future work. First, this paper focuses on unfixable constraints, while accepting partial models with fixable violations.
% but many constraints are fixable later in generation, such as requiring an \texttt{And} node to have two incoming edges in \clevr{}. 
Future work could explore how to use these fixable violations as decoding-time heuristics. 
% Second, \approach{} assumes that the DSL provides an explicit grammar, metamodel, well-formedness constraints, and a mapping between grammar and metamodel. Although these assumptions are natural for many SE applications, other settings may require manual effort to construct such specifications. Future work could investigate how to derive complete DSL specifications from incomplete ones.  
Currently, \approach{} focuses only on the final output generated by the \llm{}. Future work could extend the framework to reasoning \llm{}s, where verification may guide not only the generated DSL tokens but also intermediate reasoning traces.  %, enabling backtracking over earlier reasoning steps.
Finally, \approach{} could be integrated with DSL frameworks such as Xtext~\cite{xtext} and JetBrains MPS~\cite{jetbrains_mps} to streamline DSL development with verified LLM-based instance generation.

\subsection{Threats to validity}

\noindent
\textbf{Internal validity.}
Randomness in \llm{} outputs could affect the comparison between approaches. We mitigate this by averaging results over three random seeds and evaluating all approaches under identical decoding conditions. 
% Statistical significance is assessed using the Wilcoxon signed-rank test. 
We also report win/tie/loss counts for the alignment scores, complemented with the $p$-values of the Wilcoxon signed-rank test. 
% The single observed violation of \autoref{thm:worst_case} is caused by an implementation bug in the backend solver rather than the decoding procedure.

\noindent
\textbf{External validity.}
\approach{} currently requires token-level logit access and is therefore evaluated only on open-weight \llm{}s. 
% While the formulation itself is \llm{}-agnostic, current commercial inference APIs do not provide per-step logit access, preventing a direct implementation.
We evaluate four representative DSLs covering both structural and executable semantics. The effectiveness of \approach{} on DSLs with substantially richer semantic constraints, or much larger instances, remains to be investigated. In addition, the current implementation inherits the expressiveness of \refinery{}; however, the overall approach can be integrated with any backend verifier.

\noindent
\textbf{Construct validity.}
Our evaluation intentionally includes both alignment metrics satisfying the assumptions of \autoref{thm:worst_case} (F1* and execution accuracy) and a metric that does not (F1), allowing us to distinguish theoretical guarantees from empirical behavior. Consistency and best-of-$N$ consistency rely on the correctness of the underlying verifier. Finally, execution accuracy is an imperfect proxy for true alignment correctness because an incorrect executable DSL instance may occasionally produce the correct result by chance.
\section{Conclusion}
\label{sec:future_plan}
This paper presented \approach{}, a metamodel-guided decoding framework for generating consistent DSL instances with \llm{}s. Unlike other decoding approaches, \approach{} lifts verification to the metamodel level by mapping generated text to a partial model and checking whether it remains completable with respect to the metamodel and well-formedness constraints. Under a sound verifier, we showed that \approach{} never performs worse than grammar-constrained decoding for alignment metrics satisfying natural assumptions. Across four DSL benchmarks and six open-weight \llm{}s, \approach{} improves consistency and alignment while preserving syntactic validity and introducing only modest runtime overhead. Future work includes reducing runtime through more efficient retry strategies, extending \approach{} to fixable constraints and reasoning \llm{}s, and integrating it with DSL frameworks to streamline adoption.
% \input{sections/99.notes_from_aren}

% \clearpage

\bibliographystyle{IEEEtran}
\bibliography{biblio}
\end{document}